\documentclass[runningheads,a4paper]{llncs}

\usepackage{amssymb}
\usepackage{graphicx}

\usepackage{amsmath}

\usepackage{url}
\urldef{\mailsa}\path|deng-xt@cs.sjtu.edu.cn|
\urldef{\mailsb}\path|zhe_feng@g.harvard.edu|
\urldef{\mailsc}\path|christos@cs.berkeley.edu|

\begin{document}

\mainmatter  

\title{Power-Law Distributions in a Two-sided Market\\
and Net Neutrality}

\titlerunning{Power-Law Distributions in a Two-sided Market and Net Neutrality}

%
%
\author{Xiaotie Deng \and Zhe Feng\and Christos H. Papadimitriou}
\authorrunning{X. Deng, Z. Feng, C. H. Papadimitriou}

\institute{Department of Computer Science, Shanghai Jiao Tong University, Shanghai, China\\
School of Engineering and Applied Sciences, Harvard University, Cambridge, USA\\
Department of Computer Science, University of California, Berkeley, CA, USA\\
\mailsa;
\mailsb;
\mailsc\\
}
%
%

\toctitle{Power-Law Distributions in a Two-sided Market}
\tocauthor{and Net Neutrality}
\maketitle

\begin{abstract}
``Net neutrality'' often refers to the policy dictating that an Internet service provider (ISP) cannot charge content providers (CPs) for delivering their content to consumers.  Many past quantitative models designed to determine whether net neutrality is a good idea have been rather equivocal in their conclusions.  Here we propose a very simple two-sided market model, in which the types of the consumers and the CPs are {\em power-law distributed} --- a kind of distribution known to often arise precisely in connection with Internet-related phenomena.  We derive mostly analytical, closed-form results for several regimes:  (a) Net neutrality, (b) social optimum, (c) maximum revenue by the ISP, or (d) maximum ISP revenue under quality differentiation.  One unexpected conclusion is that (a) and (b) will differ significantly, unless average CP productivity is very high.
\end{abstract}

\section{Introduction}
The Internet is by far the world's most crucial technological artifact.  A mere quarter century after its beginning, it has emerged to become, through connecting over two billion people, the nexus of all human activity --- intellectual, social, economic --- and to satisfy, to varying and rapidly evolving degrees, humanity's thirst for information and access, communication and interaction, education and wisdom, entertainment and excitement, opportunity and publicity, let alone justice, freedom, democracy.  The Internet is also a gestalt, complex system; a novel, mysterious, and fascinating scientific object studied intensely by researchers of all colors, including computer scientists and economists.

From the point of view of economics (that is to say, efficiency and scarcity) one useful abstraction of the Internet is that of a {\em two-sided market} \cite{Arm,Ti}.  In such a market, a platform (e.g., a game console, or an operating system,  or an Internet service provider (ISP)) brings together two populations of agents: players with game developers, or users with application programmers, or, in the case of ISPs, Internet users with Internet content providers (CPs, such as Google, NYT, or Shtetl-optimized).  Two-sided markets are interesting because they can exhibit network effects and other complex externalities.  An important question is, if the two populations are passive price-takers, what is the platform policy (typically, pricing for access from both sides\footnote{Possibly negative prices: recall that in the first two examples the practice includes subsidies.}) that maximizes platform revenue, and what is the socially optimum policy?  Plus, if these two differ substantially, how should the two-sided market be optimally regulated?

In the case of the Internet, this question has been known as the {\em net neutrality debate}, see \cite{wiki,open} and the related work section for the complex history and diverse and precarious current status.  The term ``net neutrality" has been used in many different senses.  Most fundamentally, and closest to home, net neutrality is the computer science argument that the {\em end-to-end} principle in networking \cite{endtoend} implies that ISPs have no access to the content or origin of packets (as such information adds nothing to the network's ability to operate properly).  In policy, law, and economics, by ``net neutrality" one typically understands two implied consequences of the end-to-end principle, namely that ISPs cannot/should not (a) treat flows differentially depending on the originating CP; or (b) charge CPs for resource use, or for content delivery to consumers.

There is a substantial and growing literature of economic research on net neutrality, and the two subtly different interpretations of the term ``net neutrality" (a) and (b) above give rise to divergent threads within it (see the Related Work section).  Typically the models include only one ISP (even though interesting analyses of multiple ISPs exist \cite{W}) who charges (or does not) the two sides of the market for access, while the utility of the two populations is modeled in a  number of different natural ways.  Unsurprisingly, there is no definitive answer in the literature to the key question above (``which ISP policy is socially optimal?"), even though interesting points can be made based on such models  (more in the related work section).

\paragraph{The model.}  In this paper we introduce and analyze a new model of two-sided market motivated by the net neutrality problem.  Our goals in defining this model have been these:

\begin{itemize}
\item {\em Keep the model very simple,} with {very few and crisp parameters and assumptions}, so that general conclusions can be drawn.

\item At the same time, adopt assumptions (e.g. about distributions) compatible with the acknowledged reality of the Internet.  Our model is the first to assume that the types of both users and CPs are {\em power-law distributed.}
\end{itemize}

Power law distributions \cite{FFF,M} (see also \cite{Gab} for their use in economic modeling) are simple distributions outside the exponential family, with one parameter (the exponent) typically ranging between 2 and 3 --- thus, they also serve our goal of parametric parsimony.  Even though they had been observed in many places since the early 20th century (in city populations, word frequencies, incomes, etc.\footnote{Power law distributions have been called ``the signature of human activity,'' even though they also appear in life and the cosmos, and they are easy to confuse with the lognormal distribution.}), it was the Internet that brought them to the center of technical discourse --- indeed it seems almost impossible to understand and model any aspect of the Internet and the web without resorting to these types of distributions. It seems natural to suppose that CP type (capturing the CP's quality, or market share, or size) is so distributed, since power-law distributed firm size is a known characteristic of dynamic industries.  It is also natural to accept that consumer types (measuring motivation, interest in the Internet) are power-law distributed --- for example, incomes are distributed this way.  Type distributions lie at the basis of our model.  The product of consumer type $x$ times CP type $y$, times the speed of the net, captures the {\em matching probability}, the probability that a consumer of type $x$ will ``like'' (download content from) a CP of type $y$.  This, together with a simple assumption on network speed (we take it inversely proportional to total traffic) defines the expected utility of both CPs and consumers:  For a CP we assume it is proportional to the number of consumers who like it (modeling advertising income, or else popularity) and for a consumer a concave function, such as the square root, of the number of CPs that s/he likes.  Finally, in the appendix we also briefly discuss a simple model of quality-of-service differentiation where an ISP charges CPs for using a privileged channel akin to the so-called ``Paris metro pricing'' \cite{Odl}.

Naturally, there are many aspects of this complex problem that we do not model:   We do not model ISP costs, and, most importantly for the net neutrality debate, ISP technology and investment.  However, our work can inform this crucial aspect of the problem, as our analytical results depend explicitly on the total network capacity.  Our model of CP cost is simplistic (we assume that it is, in expectation, proportional to its type), but we have obtained similar results under different assumptions.  We assume that there is only one ISP (as does most of the literature); however, our results can be used to solve simple models with many ISPs.  And we do not model one of the salient characteristics of the Internet, namely its rapid growth; however, our use of power-law distributions in CP size can be seen as taking into account the exquisitely dynamic nature of the Internet market.

\paragraph{Our results.}  {\em We derive closed-form analytical results} for almost all of the questions raised by our model: For the optimum ISP policy, for the optimum ISP policy under net neutrality, as well as for the ISP policy that maximizes social welfare, but also for comparisons between them; for a few points that are hard to answer analytically, we have very clean computational results.

Our most surprising conclusion is that, in this model, {\em net neutrality is not socially optimal} unless CP costs are very small.  That is, there is in general a socially optimum price the ISP should charge the CPs, and this price is zero only if a parameter measuring CP costs (essentially, the average inverse productivity in the CP industry) is below a threshold.     Regulation is needed for efficiency, requiring the ISP to charge CPs not necessarily zero, as in net neutrality, but the socially optimum price, typically smaller than what the ISP would like to charge.

The question then, for the regime of large CP costs, becomes:  among the two suboptimal extremes (net neutrality or ISP revenue maximization), which is the more efficient?  It turns out that the answer varies (see the computational results in Fig. \ref{fig:sw}):  For CP costs just above the neutrality threshold, net neutrality is better.  For larger CP costs, ISP revenue maximization is better.  Interestingly, in both cases the differences in social welfare between the three regimes does not seem that great.  Overall, our model yields concrete, quantitative, and crisp results for the net neutrality problem, stemming from rather involved analysis, of the kind we believe had not been available in the literature, for a kind of model (consumers and CPs of power-law distributed types) that is arguably especially fit for the problem in hand.

Our results are summarized in Table \ref{tab:main}. The parameters shown in this table will be mentioned in next section.

\begin{table}[h]
	\vspace{10pt}
	\renewcommand\arraystretch{1.7}
	\centering
	\caption{Summary of results}
	\begin{tabular}{c|c|c|c}
		
		\hline
		
		& CP costs ($a$) & Optimal CP fee ($b_{opt}$) & Optimal \\
		&&& membership fee ($c_{opt}$)   \\\hline
		Max-Rev  & $a>\frac{\lambda}{2}\frac{x_0^{2-\gamma}}{\gamma-2}$  & $a$ & 0 \\
		 $(c=0)$  & $0\leq a \leq \frac{\lambda}{2}\frac{x_0^{2-\gamma}}{\gamma-2}$ &$\frac{\lambda}{\gamma-2}x_0^{2-\gamma}-a$ & 0\\\hline
		Max-Rev & $a\leq \frac{1}{2}(\frac{\gamma-2}{\gamma-1}\phi'(\bar Yx_0)+\lambda)\bar X$ & $\frac{\lambda}{\gamma-2}x_0^{2-\gamma} -a$ & $\phi(\frac{1}{\beta-2}y_0^{2-\beta}x_0)$\\
		 $(c>0)$ &$a> \frac{1}{2}(\frac{\gamma-2}{\gamma-1}\phi'(\bar Yx_0)+\lambda)\bar X$ &$(\frac{2\lambda}{\frac{\gamma-1}{\gamma-2}\phi'(\sqrt{\bar Y\frac{1}{\beta-2}y^{* 2-\beta}}x_0)+\lambda}-1)a$ &  $\phi(\sqrt{\bar Y\frac{1}{\beta-2}y^{* 2-\beta}}x_0)$\\\hline
		Socially & $a\leq \frac{1}{2}(\int_{x_0}^{\infty}\phi'(x\bar Y)x^{1-\gamma}dx+\lambda \bar X)$ & $\leq \frac{\lambda}{\gamma-2}x_0^{2-\gamma} -a$ & $\leq \phi(\sqrt{\bar Y\frac{1}{\beta-2}\widehat{y}^{2-\beta}}x_0)$ \\
		Optimum  &$a\> \frac{1}{2}(\int_{x_0}^{\infty}\phi'(x\bar Y)x^{1-\gamma}dx+\lambda \bar X)$ & $(\frac{2\lambda}{\frac{1}{\bar X}\int_{x_0}^{\infty}\phi'(\sqrt{\bar Y\frac{1}{\beta-2}\widehat{y}^{2-\beta}}x)x^{1-\gamma}dx+\lambda}-1)a$ & $\leq \phi(\sqrt{\bar Y\frac{1}{\beta-2}\widehat{y}^{2-\beta}}x_0)$\\\hline
	\end{tabular}
	\label{tab:main}
\end{table}


\subsection*{Related Work}
For aspects of policy, law, and history of the subject see \cite{open,Peha,Tate,wiki}.  \cite{LW} is an eloquent advocacy of net neutrality backed by modest quantitative argument, while \cite{Alt} is an exploration of possible business models in the CP industry and the ways they affect the net neutrality issue;  the model involves only one CP.  \cite{EH} propose a sophisticated and realistic model of CP-consumer interaction, but the complexity of their model prevents definite conclusions about net neutrality; an important monotonicity principle is shown, stating that social welfare is always coterminous with the total content transmitted through the network.  In earlier work \cite{ET}, a simple model in a similar spirit to ours was proposed, albeit with CP and consumer types uniformly distributed.  Their results are dependent on parameter value ranges, with CP costs playing an important role, as they do in our results.  In the model of \cite{W} there are many regional monopolist ISPs, and deviation from net neutrality leads in a tragedy in the commons situation (the commons being the CP industry).  The effect and nature of competition among ISPs is taken on in \cite{K}, through a mostly qualitative analysis.

Net neutrality as differentiation in quality of service has also been addressed in the economic literature.  In \cite{Cheng,Choi} consumers are connected to {\em two} CPs through a single ISP running a network with realistic (i.e., informed by queueing theory) delays, and two levels of service (a fast lane sold through bidding in \cite{Choi}, a priority service in \cite{Cheng}), and the two CPs choose level of service according to their profitability.  In contrast, \cite{Kr} models CPs by their tolerance of network delays.  Finally, \cite{Nj} model the network as a sophisticated extensive-form game, in which CPs, ISPs, and consumers interact by setting prices and choosing services; they conclude that net neutrality prevails in several environments, for example in the presence of priority lanes.



\section{The Model}
In our basic model an ISP delivers the content of CPs to a population of consumers:
\begin{itemize}
  \item The consumers are modeled as a continuum of values for the {\em consumer type} $X$, intuitively, a measure of the value this particular consumer receives from browsing the Internet.  Importantly we assume that $X$ is power-law distributed, that is, the density function is $p_\gamma(x)=x^{-\gamma}$ for $x\geq x_0$, where $x_0=(\frac{1}{\gamma-1})^{\frac{1}{\gamma-1}}$ is the minimum type.   We denote the expectation of $X$ by $\bar X = \frac{1}{\gamma-2}(\frac{1}{\gamma-1})^{\frac{2-\gamma}{\gamma-1}}$.

\item Similarly, each CP has a type $Y$ with density function $p_\beta(y)=y^{-\beta}$ for all $y\geq y_0=(\frac{1}{\beta-1})^{\frac{1}{\beta-1}}$, a measure of the CPs quality, or size.  Again, ${\bar Y} = \frac{1}{\beta-2}(\frac{1}{\beta-1})^{\frac{2-\beta}{\beta-1}}$.

\item Bandwidth and speed:  The ISP provides bandwidth $B$ ($B$ is taken to be one for simplicity, even though our results can be rewritten as functions also of $B$, for the study of issues of investment and technology innovation by the ISP).
The speed of the network is then a decreasing function of the total traffic $T$, denoted $Sp(T)$, specified next.

\item Calculation of $T$.  Crucially, we assume that the infinitesimal contribution to traffic by consumers of  type\footnote{More formally, of types between $x$ and $x+dx$, etc.} $x$ and CPs of type $y$, or equivalently, the intensity with which a consumer of type $x$ will like and download the content of a CP of type $y$, is proportional to the product of the three magnitudes $x$, $y$, and $Sp(T)$ (times $dx\cdot dy$, of course).
Therefore, the total traffic is
  $$T=\int_{x_t}^{\infty}\int_{y_t}^{\infty}Sp(T)xy p_\gamma(x)p_\beta(y)dxdy$$
Here $x_t$ and $y_t$ are the key parameters sought by our analysis, namely the minimum types of consumers and CPs respectively that participate in the market (do not drop out), given the charges imposed by the ISP.
The maximum traffic $T_0$ occurs when $x_t=x_0$ and $y_t=y_0$.
We use the relative speed function $Sp(T)=\frac{T_0}{T}$. Thus,
$$T_0=\int_{x_0}^{\infty}\int_{y_0}^{\infty}xy p_\gamma(x)p_\beta(y)dxdy = \int_{x_0}^{\infty}x^{1-\gamma}dx\int_{y_0}^{\infty}y^{1-\beta}dy= \bar X\bar Y$$

Since $T$ only depends on $x_t, y_t$,
$$T = \sqrt{\int_{x_t}^{\infty}\int_{y_t}^{\infty}T_0 xy p_\gamma(x)p_\beta(y)dxdy} = \sqrt{\bar X\bar Y\int_{x_t}^{\infty}x^{1-\gamma}dx\int_{y_t}^{\infty}y^{1-\beta}dy }$$

\item Utility functions.
\begin{itemize}
\item The utility of a user of type $x$ is assumed to be $\phi(N_x)-c$, where $N_x$ is the expected number of content providers this user likes, $c$ is the membership fee imposed on users by the ISP (independent of the traffic), and $\phi(r)$ is a concave function such as $\sqrt{r}$.  Therefore, the utility function for a user of type $x$ is
\begin{equation*}
\phi(\int_{y_t}^{\infty}\frac{T_0}{T}xy^{1-\beta}dy)-c
\end{equation*}

\item Finally, we assume that the utility function of a content provider of type $y$ is $\lambda N_{y}-by-ay$ where
\begin{itemize}
\item $\lambda$ is a needed ``exchange rate'' between the utility of consumers and that of CPs;
\item $N_{y}$ is the expected number of consumers who like this content provider --- notice that we assume advertising income to be proportional to the number of users;
\item $ay$ is the expected costs of a content provider of type $y$;
\item $by$ is the payment that the content provider needs to pay to the platform. Notice here a simplifying modeling maneuver:  While we would like to make the CP's payment a linear function of the traffic originating from it, which is roughly $N_{y}$, we make it instead a linear function of its quality $y$, which is proportional to $N_{y}$.
\end{itemize}
Thus, the utility function of a content provider with quality $y$ is as follows:
\begin{equation*}
\lambda \int_{x_t}^{\infty}\frac{T_0}{T}x^{1-\gamma}ydx -by -ay
\end{equation*}

\end{itemize}
\item Revenue of the ISP, from charges imposed on consumers and CPs:
$${\cal R}=c\int_{x_t}^{\infty}x^{-\gamma}dx+b\int_{y_t}^{\infty}y\times y^{-\beta}dy=c\int_{x_t}^{\infty}x^{-\gamma}dx+b\int_{y_t}^{\infty}y^{1-\beta}dy,$$
\item Thus, the parameters of our model are these:  power-law exponents $\gamma$ and $\beta$; the consumer concave function $\phi$; and the CP utility parameters $a$ (expected cost per unit of size) and $\lambda$. The decision variables are $b$ and $c$ (the prices charged).
\end{itemize}

\section{Revenue Maximization}

We calculate the optimum prices for the ISP to charge the two sides of the market.  For technical reasons we start by finding the optimum $b$ (CP fee) when $c=0$ (this is Theorem 1), and then proceed to the general case (Theorem 2).  The proofs are in the Appendix \ref{appendix:a}.

\begin{theorem}\label{thm:c=0,a>0}
If $c=0$, the optimal pricing strategy is
$$
b_{opt}=\left\{
\begin{array}{ccc}
  a &  & a> \frac{\lambda}{2}\frac{x_0^{2-\gamma}}{\gamma-2}\\
   \frac{\lambda}{\gamma-2}x_0^{2-\gamma}-a& & 0\leq a \leq \frac{\lambda}{2}\frac{x_0^{2-\gamma}}{\gamma-2}
\end{array}\right.$$
\end{theorem}

\begin{theorem}\label{thm:c>0, a>0}
  If $c>0, a\geq0$ and $\phi(\cdot)$ is a positive increasing concave function,  the optimal pricing strategy is

  \begin{itemize}
    \item If $a\leq \frac{1}{2}(\frac{\gamma-2}{\gamma-1}\phi'(\bar Yx_0)+\lambda)\bar X$,
          \begin{equation*}
        \left\{
        \begin{array}{c}
          b_{opt} = \frac{\lambda}{\gamma-2}x_0^{2-\gamma} -a \\
          c_{opt} = \phi(\frac{1}{\beta-2}y_0^{2-\beta}x_0)
        \end{array}\right.
      \end{equation*}
    \item If $a> \frac{1}{2}(\frac{\gamma-2}{\gamma-1}\phi'(\bar Yx_0)+\lambda)\bar X$,

    \begin{equation*}
        \left\{
        \begin{array}{c}
          b_{opt} = (\frac{2\lambda}{\frac{\gamma-1}{\gamma-2}\phi'(\sqrt{\bar Y\frac{1}{\beta-2}y^{* 2-\beta}}x_0)+\lambda}-1)a \\
          c_{opt} = \phi(\sqrt{\bar Y\frac{1}{\beta-2}y^{* 2-\beta}}x_0)
        \end{array}\right.
    \end{equation*}
    where $y^*$ is the solution of $y_t$ which satisfies the following equation:
    $$a-\frac{1}{2}(\frac{\gamma-2}{\gamma-1}\phi'(\frac{\sqrt{T_0}\sqrt{\int_{y_t}^{\infty}y^{1-\beta}dy}}{\sqrt{\int_{x_0}^{\infty}x^{1-\gamma}dx}}x_0)+\lambda)\frac{\sqrt{T_0}\sqrt{\int_{x_0}^{\infty}x^{1-\gamma}dx}}{\sqrt{\int_{y_t}^{\infty}y^{1-\beta}dy}}=0$$
  \end{itemize}

\end{theorem}

\section{Socially optimum pricing}
While excluding consumers is obviously inefficient, rather surprisingly including all CPs may not be socially optimal.  The intuitive reason is that low quality CPs clutter the Internet and incur large costs without adding enough value.  Again we must determine the optimal $x_t$ and $y_t$, and the corresponding $b$ and $c$.  Let ${\cal S}$ denote the social welfare.  We have:

\begin{equation}\label{eq:social_welfare}
  {\cal S}= \int_{x_t}^{\infty}v(x)x^{-\gamma}dx+\int_{y_t}^{\infty}v(y)y^{-\beta}dy-a\int_{y_t}^{\infty}y^{1-\beta}dy
\end{equation}

where \begin{equation}\begin{aligned}
      v(x)&=\phi(\int_{y_t}^{\infty}\frac{T_0}{T}y^{1-\beta}x dy)=\phi(\frac{\sqrt{T_0\int_{y_t}^{\infty}y^{1-\beta}dy}}{\sqrt{\int_{x_t}^{\infty}x^{1-\gamma}dx}}x)
      \end{aligned}
      \end{equation}

and \begin{equation}\begin{aligned}
    v(y) & = \lambda\int_{x_t}^{\infty}\frac{T_0}{T}x^{1-\gamma}y dx = \lambda \frac{\sqrt{T_0\int_{x_t}^{\infty}x^{1-\gamma}dx}}{\sqrt{\int_{y_t}^{\infty}y^{1-\beta}dy}}y
    \end{aligned}
    \end{equation}

Plugging these two equation above into Eq.~\ref{eq:social_welfare}, we get

\begin{equation*}
  \begin{aligned}
  {\cal S}=\int_{x_t}^{\infty}\phi(\frac{\sqrt{T_0\int_{y_t}^{\infty}y^{1-\beta}dy}}{\sqrt{\int_{x_t}^{\infty}x^{1-\gamma}dx}}x)x^{-\gamma}dx+
  \lambda\sqrt{T_0}\sqrt{\int_{x_t}^{\infty}x^{1-\gamma}dx}\sqrt{\int_{y_t}^{\infty}y^{1-\beta}dy}-a\int_{y_t}^{\infty}y^{1-\beta}dy
  \end{aligned}
\end{equation*}

We can prove the following.
\begin{theorem}\label{thm:sw1}
  To maximize the social welfare, the optimal $x_t=x_0$, while the optimal $y_t$ satisfies the following
$$y_t=\left\{
\begin{array}{ccc}
y_0 & & if a\leq \frac{1}{2}(\int_{x_0}^{\infty}\phi'(x\bar Y)x^{1-\gamma}dx+\lambda \bar X)\\
\widehat{y} & & if a> \frac{1}{2}(\int_{x_0}^{\infty}\phi'(x\bar Y)x^{1-\gamma}dx+\lambda \bar X)
\end{array}\right.
$$

where $\widehat{y}$ is the solution of the following equation of $y_t$:
$$a-\frac{1}{2}(\int_{x_0}^{\infty}\phi'(x\sqrt{\bar Y\int_{y_t}^{\infty}y^{1-\beta}dy} )x^{1-\gamma}dx+\lambda \bar X)\frac{\sqrt{\bar Y}}{\sqrt{\int_{y_t}^{\infty}y^{1-\beta}dy}})=0$$

In terms of the pricing strategy,

  \begin{itemize}
    \item If $a\leq \frac{1}{2}(\int_{x_0}^{\infty}\phi'(x\bar Y)x^{1-\gamma}dx+\lambda \bar X)$,
    \begin{equation*}
        \left\{
        \begin{array}{c}
          \widehat{b}_{opt} \leq \frac{\lambda}{\gamma-2}x_0^{2-\gamma} -a \\
          \widehat{c}_{opt} \leq \phi(\frac{1}{\beta-2}y_0^{2-\beta}x_0)
        \end{array}\right.
      \end{equation*}
    \item If $a > \frac{1}{2}(\int_{x_0}^{\infty}\phi'(x\bar Y)x^{1-\gamma}dx+\lambda \bar X)$

    \begin{equation*}
        \left\{
        \begin{array}{c}
          \widehat{b}_{opt} = (\frac{2\lambda}{\frac{1}{\bar X}\int_{x_0}^{\infty}\phi'(\sqrt{\bar Y\frac{1}{\beta-2}\widehat{y}^{2-\beta}}x)x^{1-\gamma}dx+\lambda}-1)a \\
          \widehat{c}_{opt} \leq \phi(\sqrt{\bar Y\frac{1}{\beta-2}\widehat{y}^{2-\beta}}x_0)
        \end{array}\right.
      \end{equation*}
  \end{itemize}

\end{theorem}
\begin{proof}
  Firstly, we consider the optimal $x_t$ to maximize social welfare of the platform. As we know, ${\cal S}$ is a function of $x_t$ and $y_t$, which is denoted by $S(x_t, y_t)$.
  \begin{equation}\label{eq:social_welfare_partial_x}
  \begin{aligned}
  \frac{\partial S(x_t, y_t)}{\partial x_t}&=
  \frac{1}{2}\frac{x_t^{1-\gamma}}{\int_{x_t}^{\infty}x^{1-\gamma}dx}\int_{x_t}^{\infty}\phi'(\frac{\sqrt{T_0\int_{y_t}^{\infty}y^{1-\beta}dy}}{\sqrt{\int_{x_t}^{\infty}x^{1-\gamma}dx}}x)\frac{\sqrt{T_0\int_{y_t}^{\infty}y^{1-\beta}dy}}{\sqrt{\int_{x_t}^{\infty}x^{1-\gamma}dx}}x^{1-\gamma}dx\\
  &-\phi(\frac{\sqrt{T_0\int_{y_t}^{\infty}y^{1-\beta}dy}}{\int_{x_t}^{\infty}x^{1-\gamma}dx}x_t)x_t^{-\gamma}-\frac{1}{2}\frac{\sqrt{T_0\int_{y_t}^{\infty}y^{1-\beta}dy}}{\sqrt{\int_{x_t}^{\infty}x^{1-\gamma}dx}}x_t^{1-\gamma}\\
  &\leq \frac{1}{2}\frac{x_t^{-\gamma}}{\int_{x_t}^{\infty}x^{1-\gamma}dx}\int_{x_t}^{\infty}\phi'(\frac{\sqrt{T_0\int_{y_t}^{\infty}y^{1-\beta}dy}}{\sqrt{\int_{x_t}^{\infty}x^{1-\gamma}dx}}x_t)\frac{\sqrt{T_0\int_{y_t}^{\infty}y^{1-\beta}dy}}{\sqrt{\int_{x_t}^{\infty}x^{1-\gamma}dx}}x_t x^{1-\gamma}dx\\
  &-\phi(\frac{\sqrt{T_0\int_{y_t}^{\infty}y^{1-\beta}dy}}{\int_{x_t}^{\infty}x^{1-\gamma}dx}x_t)x_t^{-\gamma}-\frac{1}{2}\frac{\sqrt{T_0\int_{y_t}^{\infty}y^{1-\beta}dy}}{\sqrt{\int_{x_t}^{\infty}x^{1-\gamma}dx}}x_t^{1-\gamma}\\
  &\leq \frac{1}{2}\frac{x_t^{-\gamma}}{\int_{x_t}^{\infty}x^{1-\gamma}dx}\phi(\frac{\sqrt{T_0\int_{y_t}^{\infty}y^{1-\beta}dy}}{\int_{x_t}^{\infty}x^{1-\gamma}dx}x_t)\int_{x_t}^{\infty}x^{1-\gamma}dx-\phi(\frac{\sqrt{T_0\int_{y_t}^{\infty}y^{1-\beta}dy}}{\int_{x_t}^{\infty}x^{1-\gamma}dx}x_t)x_t^{-\gamma}\\
  &-\frac{1}{2}\frac{\sqrt{T_0\int_{y_t}^{\infty}y^{1-\beta}dy}}{\sqrt{\int_{x_t}^{\infty}x^{1-\gamma}dx}}x_t^{1-\gamma}= -\frac{1}{2}\phi(\frac{\sqrt{T_0\int_{y_t}^{\infty}y^{1-\beta}dy}}{\int_{x_t}^{\infty}x^{1-\gamma}dx}x_t)x_t^{-\gamma}-\frac{1}{2}\frac{\sqrt{T_0\int_{y_t}^{\infty}y^{1-\beta}dy}}{\sqrt{\int_{x_t}^{\infty}x^{1-\gamma}dx}}x_t^{1-\gamma}
  \\
  &\leq 0
  \end{aligned}
  \end{equation}

The first inequality in above proof is based on the fact that $\phi'$ is a decreasing function of $x_t$. The second inequality is because $\forall x\geq 0, x\phi'(x)\leq\phi(x)$. Therefore, the optimal $x_t$ is $x_0$.

Next, we consider $y_t$.
  \begin{equation}\label{eq:social_welfare_partial_y}
  \begin{aligned}
  \frac{\partial S(x_t, y_t)}{\partial y_t}&=ay_t^{1-\beta}-\frac{1}{2}\int_{x_t}^{\infty}\phi'(\frac{\sqrt{T_0\int_{y_t}^{\infty}y^{1-\beta}dy}}{\sqrt{\int_{x_t}^{\infty}x^{1-\gamma}dx}}x)
  \frac{\sqrt{T_0}y_t^{1-\beta}}{\sqrt{\int_{x_t}^{\infty}x^{1-\gamma}dx\int_{y_t}^{\infty}y^{1-\beta}dy}}x^{1-\gamma}dx\\
  &-\frac{1}{2}\lambda\frac{\sqrt{T_0\int_{x_t}^{\infty}x^{1-\gamma}dx}}{\sqrt{\int_{y_t}y^{1-\gamma}dy}}y_t^{1-\beta}\\
  &=y_t^{1-\beta}(a-\frac{1}{2}(\int_{x_t}^{\infty}\phi'(\frac{\sqrt{T_0\int_{y_t}^{\infty}y^{1-\beta}dy}}{\sqrt{\int_{x_t}^{\infty}x^{1-\gamma}dx}}x)x^{1-\gamma}dx+\lambda\int_{x_t}^{\infty}x^{1-\gamma}dx)\frac{\sqrt{\bar Y}}{\sqrt{\int_{y_t}^{\infty}y^{1-\beta}dy}})\\
  & = y_t^{1-\beta}(a-\frac{1}{2}(\int_{x_0}^{\infty}\phi'(x\sqrt{\bar Y\int_{y_t}^{\infty}y^{1-\beta}dy} )x^{1-\gamma}dx+\lambda \bar X)\frac{\sqrt{\bar Y}}{\sqrt{\int_{y_t}^{\infty}y^{1-\beta}dy}})
  \end{aligned}
  \end{equation}

  Based on the same discussion in the proof of Theorem \ref{thm:c>0, a>0}, $$h(y_t)=\frac{1}{2}(\int_{x_0}^{\infty}\phi'(x\sqrt{\bar Y\int_{y_t}^{\infty}y^{1-\beta}dy} )x^{1-\gamma}dx+\lambda \bar X)\frac{\sqrt{\bar Y}}{\sqrt{\int_{y_t}^{\infty}y^{1-\beta}dy}})$$ 
  
  is an increasing function of $y_t$. Thus,

  \begin{itemize}
    \item If $a\leq h(y_0)$, then $\frac{\partial S(x_t, y_t)}{\partial y_t}\leq 0$. Thus the optimal $y_t$ is $y_0$.

    In this case, the optimal pricing strategy is
    \begin{equation*}
        \left\{
        \begin{array}{c}
          \widehat{b}_{opt} \leq \frac{\lambda}{\gamma-2}x_0^{2-\gamma}-a \\
          \widehat{c}_{opt} \leq \phi(\frac{1}{\beta-2}y_0^{2-\beta}x_0)
        \end{array}\right.
      \end{equation*}
    \item If $a > h(y_0)$, then there exists a unique solution $\widehat{y}$ for $h(y_t)-a=0$. Then the socially optimal pricing is

    \begin{equation*}
        \left\{
        \begin{array}{c}
          \widehat{b}_{opt} = (\frac{2\lambda}{\frac{1}{\bar X}\int_{x_0}^{\infty}\phi'(\sqrt{\bar Y\frac{1}{\beta-2}\widehat{y}^{2-\beta}}x)x^{1-\gamma}dx+\lambda}-1)a \\
          \widehat{c}_{opt} \leq \phi(\sqrt{\bar Y\frac{1}{\beta-2}\widehat{y}^{2-\beta}}x_0)
        \end{array}\right.
      \end{equation*}
  \end{itemize}

\end{proof}


\subsection{Comparison of $\widehat{y}$ and $y^*$}

We would like to know the relationship between the socially optimum cut off point for CPs $\widehat{y}$ and its revenue maximizing counterpart $y^*$. This relationship depends on $\gamma, \beta, $ and $\phi$. When $\phi$ belongs to a natural class of concave functions --- namely, fractional powers --- such comparison is possible:  Revenue maximization demands that more CPs be cut off than does efficiency, assuming CP costs are not very low.

Let us define two important constants 
$$\zeta = \max\{\frac{1}{2}(\frac{\gamma-2}{\gamma-1}\phi'(\bar Yx_0)+\lambda)\bar X, \frac{1}{2}(\int_{x_0}^{\infty}\phi'(x\bar Y)x^{1-\gamma}dx+\lambda \bar X)\}$$ 

$$\eta = \min\{\frac{1}{2}(\frac{\gamma-2}{\gamma-1}\phi'(\bar Yx_0)+\lambda)\bar X, \frac{1}{2}(\int_{x_0}^{\infty}\phi'(x\bar Y)x^{1-\gamma}dx+\lambda \bar X)\}$$

\begin{theorem}\label{thm:comparison}
  Suppose $a>\zeta$ and $\phi(x)=x^\theta$ where $0<\theta<1$.  Then $\widehat{y}<y^*$.
\end{theorem}
\begin{proof}
For $y^*$,

\begin{equation}\label{eq:optimal_renvenue}
\begin{aligned}
  a&=\frac{1}{2}(\frac{\gamma-2}{\gamma-1}\phi'(\frac{\sqrt{T_0}\sqrt{\int_{y^*}^{\infty}y^{1-\beta}dy}}{\sqrt{\int_{x_0}^{\infty}x^{1-\gamma}dx}}x_0)+\lambda)\frac{\sqrt{T_0}\sqrt{\int_{x_0}^{\infty}x^{1-\gamma}dx}}{\sqrt{\int_{y^*}^{\infty}y^{1-\beta}dy}}\\
  &=\frac{1}{2}(\frac{\gamma-2}{\gamma-1}\phi'(\sqrt{\bar Y\int_{y^*}^{\infty}y^{1-\beta}dy}x_0)+\lambda)\frac{\bar X\sqrt{\bar Y}}{\sqrt{\int_{y^*}^{\infty}y^{1-\beta}dy}}=g(y^*)
\end{aligned}
\end{equation}

For $\widehat{y}$,

\begin{equation}\label{eq:optimal_social_welfare}
a=\frac{1}{2}(\int_{x_0}^{\infty}\phi'(x\sqrt{\bar Y\int_{\widehat{y}}^{\infty}y^{1-\beta}dy})x^{1-\gamma}dx+\lambda \bar X)\frac{\sqrt{\bar Y}}{\sqrt{\int_{\widehat{y}}^{\infty}y^{1-\beta}dy}})=h(\widehat{y})
\end{equation}

Suppose $\widehat{y}=y^*=y'$, then

\begin{equation}
\begin{aligned}
&h(y')-g(y')\\
&=\frac{1}{2}\frac{\sqrt{\bar Y}}{\sqrt{\int_{y'}^{\infty}y^{1-\beta}dy}}y'(\int_{x_0}^{\infty}\phi'(x\sqrt{\bar Y\int_{y'}^{\infty}y^{1-\beta}dy})x^{1-\gamma}dx-\frac{\gamma-2}{\gamma-1}\phi'(x_0\sqrt{\bar Y\int_{y'}^{\infty}y^{1-\beta}dy})\bar X)\\
\end{aligned}
\end{equation}

Let $\sqrt{\bar Y\int_{y'}^{\infty}y^{1-\beta}dy}=Z$, then

\begin{equation}
\begin{aligned}
&\int_{x_0}^{\infty}\phi'(Zx)x^{1-\gamma}dx-\frac{\gamma-2}{\gamma-1}\phi'(Zx_0)\bar X\\
& = \int_{x_0}^{\infty}\theta Z^{\theta-1}x^{\theta-\gamma}dx-\frac{\gamma-2}{\gamma-1}\theta(Zx_0)^{\theta-1}\frac{1}{\gamma-2}x_0^{2-\gamma}\\
& = \theta Z^{\theta-1}x_0^{1+\theta-\gamma}(\frac{1}{\gamma-\theta-1}-\frac{1}{\gamma-1})>0
\end{aligned}
\end{equation}

Thus, if $\widehat{y}=y^*$, $h(\widehat{y})-g(y^*)>0$. Since $g$ and $h$ are both increasing functions, then $\widehat{y}<y^*$ if $h(\widehat{y})=g(y^*)=a$.
\end{proof}

\subsection{Welfare Comparison}
To summarize our results so far:
\begin{itemize}
  \item In both revenue and welfare maximization, no consumers are left outside the market.
    \item When CP costs are small 
($a\leq \eta$), then no CPs are cut off either.
 \item When $\eta<a\leq\zeta$,  then no CPs are cut off for social optimality, however, some CPs will be cut off for revenue optimality.\footnote{This is because $\frac{1}{2}(\frac{\gamma-2}{\gamma-1}\phi'(\bar Yx_0)+\lambda)\bar X < \frac{1}{2}(\int_{x_0}^{\infty}\phi'(x\bar Y)x^{1-\gamma}dx+\lambda \bar X)$ when $\phi$ is a fractional power function.}
  \item But otherwise, some CPs must be cut off for efficiency (that is, network neutrality is socially suboptimal), while more will have to be cut off for revenue optimality.
\footnote{Perhaps what is most striking in this figure (especially to somebody trained in approximation algorithms and the price of anarchy) is that, in all three cases and for these parameters and model, neither of the two extreme regimes (revenue maximization and net neutrality) is catastrophically suboptimal in social welfare.}

\end{itemize}

But the question now arises, how does the social welfare of net neutrality compare with that of revenue optimality?  Simulations show that the answer depends on CP costs, that is to say, $a$.
In the simulation $\gamma=\beta=2.5$ and $\lambda = 0.1$ where $\phi(x)=x^{1/2}$.   Fig.~\ref{fig:sw} shows the social welfare curve for three different values of $a$:  $1.1\zeta, 1.5\zeta, 2\zeta$.   When $a=1.1\zeta$ (that is, close to the neutrality region) net neutrality has better social welfare than  revenue, while when $a=1.5\zeta, 2\zeta$ the social welfare in revenue optimum case is quite a bit larger than the social welfare in net neutrality.  In fact, we can show that there is a single transition in this regard (proof in the Appendix B):
\begin{theorem}
If $\phi(x)=x^\theta (0<\theta<1)$, there exists a unique $\bar{a}$ such that when $a< \bar a$ net neutrality has better welfare than revenue maximization, while the opposite happens when $a > \bar a$.
\end{theorem}

\begin{figure}
\centerline{\includegraphics[width=1\textwidth]{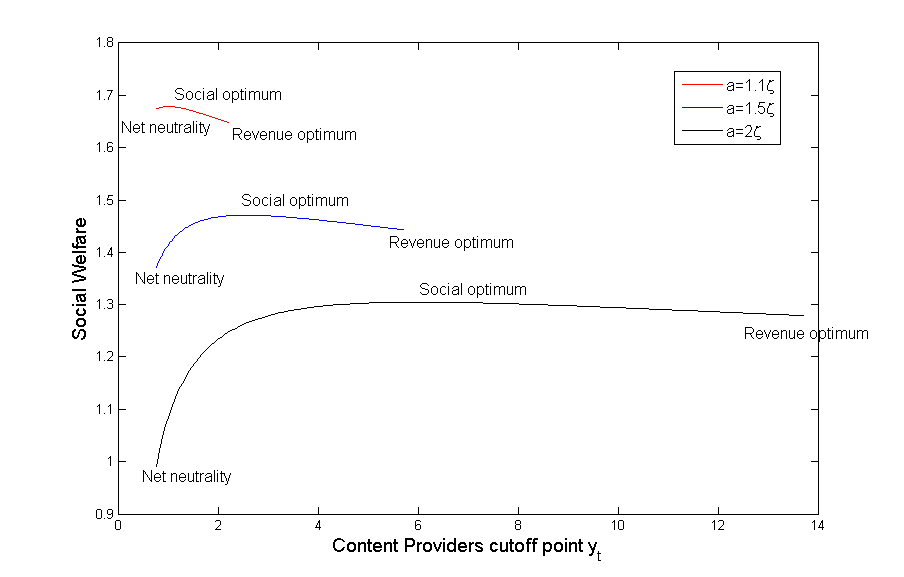}}
\caption{Social Welfare Curve}
\label{fig:sw}
\end{figure}

\section{Conclusion and Further Work}
We have presented a parsimonious model of the Internet as a two-sided market with power-law distributed types from the two sides, with a simple cost structure for CPs, and utilities for the two sides based on simple and natural assumptions.
\begin{itemize}
\item  Net neutrality is socially optimum only when CP productivity is very high.  For lower levels of CP productivity (larger $a$), net neutrality is better than ISP revenue maximization, but net neutrality is worse than ISP revenue maximization for even lower values.  The preeminence of CP productivity as the determining factor of the optimum regulatory regime is one interesting insight from our model.
\end{itemize}

There are many possible extensions that seem very interesting, and some of them appear to be within reach.
\begin{enumerate}

\item  Here we have adopted the ``no payment'' interpretation of net neutrality.  What about the ``non-differentiation'' point of view?  We have interesting preliminary results of this sort (see the Appendix C).  Assume that part of the bandwidth is set aside for paying CPs.  The point is that the payment counter-incentive will increase speed in this ``channel'' (this is the {\em Paris metro pricing} idea \cite{Odl}).  How large part of the total bandwidth should be so allocated, and how should it be priced?   In the appendix we answer these questions, analytically and in more detail computationally, for the case $a=0$  --- that is, zero costs for CPs.  The general $a$ case seems harder, but it would be interesting to crack it.

\item How could we make our model more realistic, without sacrificing much of its simplicity?  We have tried other forms of CP costs and charges (for example, constant instead of linear in $y$) without seeing qualitatively different results.  But how about changing the utility model?  One alternative model would weigh CP revenue by the type of the users it attracts.  Another would use more elaborate and realistic speed functions, for example from queueing theory.

\item  We have not considered {\em subsidies} of CPs by the ISP (negative $b$; note that subsidies are common in two-way markets).  Would they ever improve social welfare, or even ISP revenue?

\item A common argument against net neutrality is that it does not incentivize ISPs to invest in network technology. What can our model tell us about this?  In our calculations we have used, for simplicity that the bandwidth $B$ is one.  We suspect that re-introducing $B$ into our formulas might reveal interesting insights about incentives of the ISP to invest.

\item We have assumed a monopolist ISP; how would ISP competition affect the market?  We suspect that many ISPs competing for consumers under revenue maximization would result in $c=0$, and would charge CPs in near-identical ways, because each of them will be ``selling" to the CPs a different lot (in expectation of the same size) of the same product: the consumers who (randomly) chose this ISP.  Hence we suspect that our results summarized in Theorem 1 come close to obtaining yet another interesting comparison point, telling us how the ISP's monopoly is affecting the efficiency of the market.

\end{enumerate}

\section*{Acknowledgement}

Research results reported in this work are partially supported by the National Natural Science Foundation of China (Grant No. 11426026, 61632017, 61173011) and a Project 985 grant of Shanghai Jiao Tong University. This work was done while the authors was visiting the Simons Institute for the Theory of Computing, Berkeley.

\newpage
\appendix

\section*{Appendix}

\section{Proofs of Theorems 1 and 2}\label{appendix:a}

{\bf Theorem 1.} \emph{If $c=0$, the optimal pricing strategy is}
$$
b_{opt}=\left\{
\begin{array}{ccc}
a &  & a> \frac{\lambda}{2}\frac{x_0^{2-\gamma}}{\gamma-2}\\
\frac{\lambda}{\gamma-2}x_0^{2-\gamma}-a& & 0\leq a \leq \frac{\lambda}{2}\frac{x_0^{2-\gamma}}{\gamma-2}
\end{array}\right.$$

\begin{proof}
	Since $c=0$, $x_t=x_0$, and the revenue of the ISP is ${\cal R}=b\int_{y_t}^{\infty}y^{1-\beta}dy$. Since the ISP would like to maximize revenue, $by_t$ should be the threshold payment for content providers. In other words,
	
	\begin{equation}\label{eq:payment}
	b=\lambda\int_{x_0}^{\infty}\frac{T_0}{T}x^{1-\gamma}dx - a
	\end{equation}
	
	which means the utility of the content provider with $y_t$ quality is 0.
	
	Therefore, the revenue of the platform is only a function of $y_t$. We rewrite it as
	$$R(y_t)=(\lambda\int_{x_0}^{\infty}\frac{T_0}{T}x^{1-\gamma} dx - a)\int_{y_t}^{\infty}y^{1-\beta}dy$$
	
	We plug $T=\sqrt{T_0\int_{x_0}^{\infty}x^{1-\gamma}dx\int_{y_t}^{\infty}y^{1-\beta}dy}$ into the above equation,
	
	\begin{equation}\label{eq:revenue_no membership fee}
	\begin{aligned}
	R(y_t) & = (\lambda \int_{x_0}^{\infty}\frac{\sqrt{T_0}}{\sqrt{\int_{x_0}^{\infty}x^{1-\gamma}dx\int_{y_t}^{\infty}y^{1-\beta}dy}}x^{1-\gamma} dx - a)\int_{y_t}^{\infty}y^{1-\beta}dy\\
	& = (\lambda \frac{\sqrt{T_0\int_{x_0}^{\infty}x^{1-\gamma}dx}}{\sqrt{\int_{y_t}^{\infty}y^{1-\beta}dy}} - a)\int_{y_t}^{\infty}y^{1-\beta}dy\\
	& = \lambda \sqrt{T_0}\sqrt{\int_{x_0}^{\infty}x^{1-\gamma}dx}\sqrt{\int_{y_t}^{\infty}y^{1-\beta}dy} - a\int_{y_t}^{\infty}y^{1-\beta}dy
	\end{aligned}
	\end{equation}
	
	Thus,
	\begin{equation*}
	\begin{aligned}
	\frac{d R(y_t)}{d y_t}&=\lambda \sqrt{T_0}\sqrt{\int_{x_0}^{\infty}x^{1-\gamma}dx} \frac{-\frac{1}{2}y_t^{1-\beta}}{\sqrt{\int_{y_t}^{\infty}y^{1-\beta}dy}} + ay_t^{1-\beta}\\
	& = y_t^{1-\beta}(a-\frac{\lambda}{2} \sqrt{T_0}\sqrt{\int_{x_0}^{\infty}x^{1-\gamma}dx} \frac{1}{\sqrt{\int_{y_t}^{\infty}y^{1-\beta}dy}})
	\end{aligned}
	\end{equation*}
	
	Obviously, $\frac{1}{\sqrt{\int_{y_t}^{\infty}y^{1-\beta}dy}}$ is an increasing function with respect to $y_t$. Therefore, we have the following claims:
	\begin{itemize}
		\item If $a \leq \frac{\lambda}{2}\sqrt{T_0}\sqrt{\int_{x_0}^{\infty}x^{1-\gamma}dx} \frac{1}{\sqrt{\int_{y_0}^{\infty}y^{1-\beta}dy}}=\frac{\lambda}{2}\frac{x_0^{2-\gamma}}{\gamma-2}$, $\frac{d R(y_t)}{d y_t}\leq 0$. Then the $y_t=y_0$ based on the optimal pricing strategy. Based on simple calculation, $b_{opt}=\frac{\lambda}{\gamma-2}x_0^{2-\gamma}-a$.
		\item If $a >\frac{\lambda}{2}\frac{x_0^{2-\gamma}}{\gamma-2}$,
		we let $\frac{d R(y_t)}{d y_t}=0$, which implies
		\begin{equation}\label{eq:optimal_payment}
		a = \frac{\lambda}{2} \sqrt{T_0}\sqrt{\int_{x_0}^{\infty}x^{1-\gamma}dx} \frac{1}{\sqrt{\int_{y_t}^{\infty}y^{1-\beta}dy}}
		\end{equation}
		
		Based on Eq.\ref{eq:payment} and Eq.\ref{eq:optimal_payment}, the optimal payment is
		
		$$b_{opt} = a$$
		
	\end{itemize}
\end{proof}

\noindent{\bf Theorem 2. } \emph{If $c>0, a\geq0$ and $\phi(\cdot)$ is a positive increasing concave function,  the optimal pricing strategy is}

\begin{itemize}
	\item \emph{If $a\leq \frac{1}{2}(\frac{\gamma-2}{\gamma-1}\phi'(\bar Yx_0)+\lambda)\bar X$,}
	\begin{equation*}
	\left\{
	\begin{array}{c}
	b_{opt} = \frac{\lambda}{\gamma-2}x_0^{2-\gamma} -a \\
	c_{opt} = \phi(\frac{1}{\beta-2}y_0^{2-\beta}x_0)
	\end{array}\right.
	\end{equation*}
	\item \emph{If $a> \frac{1}{2}(\frac{\gamma-2}{\gamma-1}\phi'(\bar Yx_0)+\lambda)\bar X$,}
	
	\begin{equation*}
	\left\{
	\begin{array}{c}
	b_{opt} = (\frac{2\lambda}{\frac{\gamma-1}{\gamma-2}\phi'(\sqrt{\bar Y\frac{1}{\beta-2}y^{* 2-\beta}}x_0)+\lambda}-1)a \\
	c_{opt} = \phi(\sqrt{\bar Y\frac{1}{\beta-2}y^{* 2-\beta}}x_0)
	\end{array}\right.
	\end{equation*}
	\emph{where $y^*$ is the solution of $y_t$ which satisfies the following equation:}
	$$a-\frac{1}{2}(\frac{\gamma-2}{\gamma-1}\phi'(\frac{\sqrt{T_0}\sqrt{\int_{y_t}^{\infty}y^{1-\beta}dy}}{\sqrt{\int_{x_0}^{\infty}x^{1-\gamma}dx}}x_0)+\lambda)\frac{\sqrt{T_0}\sqrt{\int_{x_0}^{\infty}x^{1-\gamma}dx}}{\sqrt{\int_{y_t}^{\infty}y^{1-\beta}dy}}=0$$
\end{itemize}

\begin{proof}
	The revenue of the platform is ${\cal R}=c\int_{x_t}^{\infty}x^{-\gamma}dx + b\int_{y_t}^{\infty}y^{1-\beta}dy$. To maximize the revenue, the platform organizer would set $c$ and $b$ as the threshold fee for users and content providers. Specifically,
	
	\begin{equation}\label{eq:threshold_fee}
	\begin{aligned}
	c = \phi(\int_{y_t}^{\infty}\frac{T_0}{T}y^{1-\beta}x_t dy)\\
	b = \lambda\int_{x_t}^{\infty}\frac{T_0}{T}x^{1-\gamma} dx -a\\
	\end{aligned}
	\end{equation}
	
	where the utility of the content provider with $y_t$ quality is 0 and the utility of the user with $x_t$ quality is 0. Thus, the revenue of the platform is the function of $x_t, y_t$ by plugging $T$ into.
	
	\begin{equation}
	\begin{aligned}
	R(x_t, y_t)&=\phi(\int_{y_t}^{\infty}\frac{T_0}{T}y^{1-\beta}x_t dy)\int_{x_t}^{\infty}x^{-\gamma}dx+(\lambda\int_{x_t}^{\infty}\frac{T_0}{T}x^{1-\gamma} dx -a)\int_{y_t}^{\infty}y^{1-\beta}dy\\
	& =\phi(\int_{y_t}^{\infty}\frac{\sqrt{T_0}}{\sqrt{\int_{x_t}^{\infty}x^{1-\gamma}dx\int_{y_t}^{\infty}y^{1-\beta}dy}}y^{1-\beta}x_tdy)\int_{x_t}^{\infty}x^{-\gamma}dx+(\lambda \frac{\sqrt{T_0\int_{x_t}^{\infty}x^{1-\gamma}dx}}{\sqrt{\int_{y_t}^{\infty}y^{1-\beta}dy}} - a)\int_{y_t}^{\infty}y^{1-\beta}dy\\
	& = \phi(\frac{\sqrt{T_0\int_{y_t}^{\infty}y^{1-\beta}dy}}{\sqrt{\int_{x_t}^{\infty}x^{1-\gamma}dx}}x_t)\int_{x_t}^{\infty}x^{-\gamma}dx+(\lambda \frac{\sqrt{T_0\int_{x_t}^{\infty}x^{1-\gamma}dx}}{\sqrt{\int_{y_t}^{\infty}y^{1-\beta}dy}} - a)\int_{y_t}^{\infty}y^{1-\beta}dy\\
	& = \phi(\frac{\sqrt{T_0}\sqrt{\int_{y_t}^{\infty}y^{1-\beta}dy}}{\sqrt{\int_{x_t}^{\infty}x^{1-\gamma}dx}}x_t)\int_{x_t}^{\infty}x^{-\gamma}dx
	+\lambda\sqrt{T_0}\sqrt{\int_{x_t}^{\infty}x^{1-\gamma}dx}\sqrt{\int_{y_t}^{\infty}y^{1-\beta}dy} - a\int_{y_t}^{\infty}y^{1-\beta}dy
	\end{aligned}
	\end{equation}
	
	Thus,
	\begin{equation}\label{eq:partial_x}
	\begin{aligned}
	\frac{\partial R(x_t, y_t)}{\partial x_t}&=\int_{x_t}^{\infty}x^{-\gamma}dx \cdot\phi'(\frac{\sqrt{T_0}\sqrt{\int_{y_t}^{\infty}y^{1-\beta}dy}}{\sqrt{\int_{x_t}^{\infty}x^{1-\gamma}dx}}x_t)\sqrt{T_0}\sqrt{\int_{y_t}^{\infty}y^{1-\beta}dy}\frac{\sqrt{\int_{x_t}^{\infty}x^{1-\gamma}dx}+\frac{1}{2}\frac{x_t^{2-\gamma}}{\sqrt{\int_{x_t}^{\infty}x^{1-\gamma}dx}}}{\int_{x_t}^{\infty}x^{1-\gamma}dx}\\
	&-\phi(\frac{\sqrt{T_0}\sqrt{\int_{y_t}^{\infty}y^{1-\beta}dy}}{\sqrt{\int_{x_t}^{\infty}x^{1-\gamma}dx}}x_t)x_t^{-\gamma}-\frac{\lambda}{2}\frac{\sqrt{T_0}\sqrt{\int_{y_t}^{\infty}y^{1-\beta}dy}}{\sqrt{\int_{x_t}^{\infty}x^{1-\gamma}dx}}x_t^{1-\gamma}\\
	& = \frac{1}{\gamma-1}x_t^{1-\gamma}\cdot\phi'(\frac{\sqrt{T_0}\sqrt{\int_{y_t}^{\infty}y^{1-\beta}dy}}{\sqrt{\int_{x_t}^{\infty}x^{1-\gamma}dx}}x_t)\frac{\sqrt{T_0}\sqrt{\int_{y_t}^{\infty}y^{1-\beta}dy}}{\sqrt{\int_{x_t}^{\infty}x^{1-\gamma}dx}}(1+\frac{\gamma-2}{2})\\
	&-\phi(\frac{\sqrt{T_0}\sqrt{\int_{y_t}^{\infty}y^{1-\beta}dy}}{\sqrt{\int_{x_t}^{\infty}x^{1-\gamma}dx}}x_t)x_t^{-\gamma}-\frac{\lambda}{2}\frac{\sqrt{T_0}\sqrt{\int_{y_t}^{\infty}y^{1-\beta}dy}}{\sqrt{\int_{x_t}^{\infty}x^{1-\gamma}dx}}x_t^{1-\gamma}\\
	&\leq \frac{\gamma x_t^{-\gamma}}{2(\gamma-1)}\phi(\frac{\sqrt{T_0}\sqrt{\int_{y_t}^{\infty}y^{1-\beta}dy}}{\sqrt{\int_{x_t}^{\infty}x^{1-\gamma}dx}}x_t)-\phi(\frac{\sqrt{T_0}\sqrt{\int_{y_t}^{\infty}y^{1-\beta}dy}}{\sqrt{\int_{x_t}^{\infty}x^{1-\gamma}dx}}x_t)x_t^{-\gamma}-\frac{\lambda}{2}\frac{x_t^{1-\gamma}\sqrt{T_0}\sqrt{\int_{y_t}^{\infty}y^{1-\beta}dy}}{\sqrt{\int_{x_t}^{\infty}x^{1-\gamma}dx}}\\
	&= \phi(\frac{\sqrt{T_0}\sqrt{\int_{y_t}^{\infty}y^{1-\beta}dy}}{\sqrt{\int_{x_t}^{\infty}x^{1-\gamma}dx}}x_t)x_t^{-\gamma}(\frac{\gamma}{2(\gamma-1)}-1)-\frac{\lambda}{2}\frac{\sqrt{T_0}\sqrt{\int_{y_t}^{\infty}y^{1-\beta}dy}}{\sqrt{\int_{x_t}^{\infty}x^{1-\gamma}dx}}x_t^{1-\gamma}\\
	&\leq 0
	\end{aligned}
	\end{equation}
	
	The first inequality is based on the following property of $\phi(\cdot)$: for any $x\geq 0, x\phi'(x)\leq \phi(x)$. This is because:
	
	Let $f(x)=x\phi'(x)-\phi(x)$, then $f(0)=-\phi(0)\leq 0$ and $f'(x)=\phi'(x)+x\phi''(x)-\phi'(x)=x\phi''(x)\leq 0$. Thus, $x\phi'(x)-\phi(x)\leq -\phi(0)\leq 0$.
	
	According to the discussion above, we could know $x_t = x_0$ under the optimal pricing strategy. Then we turn to consider the optimal $y_t$ in the following,
	
	\begin{equation}\label{eq:partial_y}
	\begin{aligned}
	\frac{\partial R(x_t, y_t)}{\partial y_t}
	&=-\phi'(\frac{\sqrt{T_0}\sqrt{\int_{y_t}^{\infty}y^{1-\beta}dy}}{\sqrt{\int_{x_t}^{\infty}x^{1-\gamma}dx}}x_t)\frac{\frac{1}{2}\sqrt{T_0}y_t^{1-\beta}}{\sqrt{\int_{x_t}^{\infty}x^{1-\gamma}dx}\sqrt{\int_{y_t}^{\infty}y^{1-\beta}dy}}x_t\int_{x_t}^{\infty}x^{-\gamma}dx\\
	&-\frac{\lambda}{2}\frac{\sqrt{T_0}\sqrt{\int_{x_t}^{\infty}x^{1-\gamma}dx}y_t^{1-\beta}}{\sqrt{\int_{y_t}^{\infty}y^{1-\beta}dy}}+a y_t^{1-\beta}\\
	& = a y_t^{1-\beta}- \frac{1}{2}\frac{\sqrt{T_0}\sqrt{\int_{x_t}^{\infty}x^{1-\gamma}dx}}{\sqrt{\int_{y_t}^{\infty}y^{1-\beta}dy}}y_t^{1-\beta}(\frac{\gamma-2}{\gamma-1}\phi'(\frac{\sqrt{T_0}\sqrt{\int_{y_t}^{\infty}y^{1-\beta}dy}}{\sqrt{\int_{x_t}^{\infty}x^{1-\gamma}dx}}x_t)+\lambda)\\
	& = y_t^{1-\beta}(a-\frac{1}{2}(\frac{\gamma-2}{\gamma-1}\phi'(\frac{\sqrt{T_0}\sqrt{\int_{y_t}^{\infty}y^{1-\beta}dy}}{\sqrt{\int_{x_t}^{\infty}x^{1-\gamma}dx}}x_t)+\lambda)\frac{\sqrt{T_0}\sqrt{\int_{x_t}^{\infty}x^{1-\gamma}dx}}{\sqrt{\int_{y_t}^{\infty}y^{1-\beta}dy}})
	\end{aligned}
	\end{equation}
	
	Based on Eq. \ref{eq:partial_x}, $x_t=x_0$. Therefore, we let
	\begin{equation*}\begin{aligned}
	g(y_t)&=\frac{1}{2}(\frac{\gamma-2}{\gamma-1}\phi'(\frac{\sqrt{T_0}\sqrt{\int_{y_t}^{\infty}y^{1-\beta}dy}}{\sqrt{\int_{x_t}^{\infty}x^{1-\gamma}dx}}x_t)+\lambda)\frac{\sqrt{T_0}\sqrt{\int_{x_t}^{\infty}x^{1-\gamma}dx}}{\sqrt{\int_{y_t}^{\infty}y^{1-\beta}dy}}\\
	& =\frac{1}{2}(\frac{\gamma-2}{\gamma-1}\phi'(\frac{\sqrt{T_0}\sqrt{\int_{y_t}^{\infty}y^{1-\beta}dy}}{\sqrt{\int_{x_0}^{\infty}x^{1-\gamma}dx}}x_0)+\lambda)\frac{\sqrt{T_0}\sqrt{\int_{x_0}^{\infty}x^{1-\gamma}dx}}{\sqrt{\int_{y_t}^{\infty}y^{1-\beta}dy}}\\
	\end{aligned}
	\end{equation*}
	Since $\phi'(\cdot)$ is a decreasing function and $\frac{\sqrt{T_0}\sqrt{\int_{y_t}^{\infty}y^{1-\beta}dy}}{\sqrt{\int_{x_0}^{\infty}x^{1-\gamma}dx}}x_0$ is also a decreasing function of $y_t$, $\phi'(\frac{\sqrt{T_0}\sqrt{\int_{y_t}^{\infty}y^{1-\beta}dy}}{\sqrt{\int_{x_0}^{\infty}x^{1-\gamma}dx}}x_0)$ is an increasing function with respect to $y_t$. Besides, we notice that $\frac{1}{\sqrt{\int_{y_t}^{\infty}y^{1-\beta}dy}}$ is an increasing function of $y_t$. Thus, $g(y_t)$ is an increasing function with respect to $y_t$.
	
	\begin{itemize}
		\item If $a\leq g(y_0)=\frac{1}{2}(\frac{\gamma-2}{\gamma-1}\phi'(\bar Yx_0)+\lambda)\bar X$, then $\frac{\partial R(x_t, y_t)}{\partial y_t}\leq 0$. Then, the optimal $y_t$ is $y_0$. The optimal pricing strategy of the platform is
		\begin{equation*}
		\left\{
		\begin{array}{c}
		b_{opt} = \frac{\lambda}{\gamma-2}x_0^{2-\gamma} - a \\
		c_{opt} = \phi(\frac{1}{\beta-2}y_0^{2-\beta}x_0)
		\end{array}\right.
		\end{equation*}
		
		\item If $a> g(y_0)=\frac{1}{2}(\frac{\gamma-2}{\gamma-1}\phi'(\bar Yx_0)+\lambda)\bar X$, then there exists a unique $y_t$ s.t. $a=g(y_t)$, we denote this solution as $y^*$. The optimal pricing strategy of the platform is
		\begin{equation*}
		\left\{
		\begin{array}{c}
		b_{opt} = (\frac{2\lambda}{\frac{\gamma-2}{\gamma-1}\phi'(\sqrt{\bar Y\frac{1}{\beta-2}y^{* 2-\beta}}x_0)+\lambda}-1)a \\
		c_{opt} = \phi(\sqrt{\bar Y\frac{1}{\beta-2}y^{* 2-\beta}}x_0)
		\end{array}\right.
		\end{equation*}
		
	\end{itemize}
	
\end{proof}

\section{Proof of Theorem 5}
{\bf Theorem 5. }
\emph{If $\phi(x)=x^\theta (0<\theta<1)$, there exists a unique $\bar{a}$ such that when $a< \bar a$ net neutrality has better welfare than revenue maximization, while the opposite happens when $a \geq \bar a$.}

\begin{proof}
	From Eq \ref{eq:optimal_renvenue} it is easy to check that $a$ is an increasing function of $y^*$ since $\phi'(\cdot)$ is also an increasing function of $y^*$.
	
	Now, in revenue optimization, the social welfare ${\cal S}$ is a function of $y^*$, that is,
	\begin{equation*}
	\begin{aligned}
	S(x_0, y^*)= \int_{x_0}^{\infty}\phi(\frac{\sqrt{T_0\int_{y^*}^{\infty}y^{1-\beta}dy}}{\sqrt{\int_{x_0}^{\infty}x^{1-\gamma}dx}}x)x^{-\gamma}dx+
	\lambda\sqrt{T_0}\sqrt{\int_{x_0}^{\infty}x^{1-\gamma}dx}\sqrt{\int_{y^*}^{\infty}y^{1-\beta}dy}-a\int_{y^*}^{\infty}y^{1-\beta}dy
	\end{aligned}
	\end{equation*}
	
	Following the same computation as Eq \ref{eq:social_welfare_partial_y}, we have
	
	\begin{equation*}
	\begin{aligned}
	\frac{\partial {\cal S}}{\partial y^*} = y^{*1-\beta}(a-\frac{1}{2}(\int_{x_0}^{\infty}\phi'(x\sqrt{\bar Y\int_{y^*}^{\infty}y^{1-\beta}dy} )x^{1-\gamma}dx+\lambda \bar X)\frac{\sqrt{\bar Y}}{\sqrt{\int_{y^*}^{\infty}y^{1-\beta}dy}}))
	\end{aligned}
	\end{equation*}
	
	Based on the discussion in Theorem \ref{thm:comparison}, $a=g(y^*)<h(y^*)$. Therefore, $ \frac{\partial {\cal S}}{\partial y^*}< 0$. Thus, the social welfare ${\cal S}$ is a decreasing function of $y^*$, which implies it is a decreasing function of $a$, denoted by $S^*(a)$.
	
	In the initial $y_0$ case, the social welfare is also a decreasing function of $a$, denoted by $S_0(a)$. When $a=\zeta=\frac{1}{2}(\int_{x_0}^{\infty}\phi'(x\bar Y)x^{1-\gamma}dx+\lambda \bar X)$, $\widehat{y}=y_0<y^*$, therefore, $S_0(a) > S^*(a)$.
	
	When $a$ tends to infinity,
	\begin{equation}
	\begin{aligned}
	S^*(a) - S_0(a)& = a(\int_{y_0}^{\infty}y^{1-\beta}dy-\int_{y^*}^{\infty}y^{1-\beta}dy) - \lambda\sqrt{T_0}\sqrt{\int_{x_0}^{\infty}x^{1-\gamma}dx}(\sqrt{\int_{y_0}^{\infty}y^{1-\beta}dy}-\sqrt{\int_{y^*}^{\infty}y^{1-\beta}dy})\\
	& -\int_{x_0}^{\infty}(\phi(\frac{\sqrt{T_0\int_{y_0}^{\infty}y^{1-\beta}dy}}{\sqrt{\int_{x_0}^{\infty}x^{1-\gamma}dx}}x)-\phi(\frac{\sqrt{T_0\int_{y^*}^{\infty}y^{1-\beta}dy}}{\sqrt{\int_{x_0}^{\infty}x^{1-\gamma}dx}}x))x^{-\gamma}dx\\
	&= a(\int_{y_0}^{\infty}y^{1-\beta}dy-\int_{y^*}^{\infty}y^{1-\beta}dy) - \mathbf{neg} > 0
	\end{aligned}
	\end{equation}
	
	since the negative part of $S^*(a) - S_0(a)$, $\mathbf{neg}$ is bounded. Finally, combining the monotone property of $S_0(a)$ and $S^*(a)$ and the above, we conclude that there exists a unique transition with $S_0(a)=S^*(a)$.
	
\end{proof}

\section{Paris Metro Pricing}
Can we extend our analysis to study net neutrality under its interpretation of ``non-differentiation''?  In this section we analyze a simple model of this sort.  Suppose that the ISP splits the pipe, assumed to be of capacity one, into two channels, the {\em pay channel} of capacity $B_1\leq 1$, and the {\em free channel} with capacity $1-B_1$.  CPs can choose between the two (we assume the decision is binary).  If a CP of type $y$ uses the free chanel it is charged zero, but if it uses the paying channel it is charged $by$.  Consumers have access to content transmitted through both channels.

We can now calculate the speed at the two channels. Thus the speed function of pay channel is
$\frac{B_1 T_0}{T_1}$ where $T_1$ is the traffic of the payhannel, which is measured as follows:

\begin{equation*}\begin{aligned}
T_1&=\int_{x_t}^{\infty}\int_{y_t}^{\infty}\frac{B_1T_0}{T_1}xydp_\gamma(x)dp_\beta(y)=\frac{B_1T_0}{T_1}\int_{x_t}^{\infty}\int_{y_t}^{\infty}x^{1-\gamma}y^{1-\beta}dxdy\\
& = \sqrt{B_1T_0\int_{x_t}^{\infty}x^{1-\gamma}dx\int_{y_t}^{\infty}y^{1-\beta}dy}
\end{aligned}\end{equation*}

where $y_t$ is the threshold type of content providers for inclusion in the paying channel.

We assume that $a=0$, and hence no content providers drop out completely.

Similarly, the traffic of Channel II $T_2$ is denoted as $T_2=\sqrt{(1-B_1)T_0\int_{x_t}^{\infty}x^{1-\gamma}dx\int_{y_0}^{y_t}y^{1-\beta}dy}$

To maximize revenue,
\begin{equation}\label{membershipfee}
c = \phi((\int_{y_0}^{y_t}\frac{(1-B_1)T_0}{T_2}y^{1-\beta}dy+\int_{y_t}^{\infty}\frac{B_1T_0}{T_1}y^{1-\beta}dy)x_t)
\end{equation}

\begin{equation}\label{payment}
\lambda y_t\int_{x_t}^{\infty}\frac{B_1T_0}{T_1}x^{1-\gamma}dx - by_t = \lambda y_t\int_{x_t}^{\infty}\frac{(1-B_1)T_0}{T_2}x^{1-\gamma}dx
\end{equation}

Hence this decision problem is equivalent to the following optimization problem:

\begin{equation}\label{optimization}
\max_{x_t, y_t, B_1} R(x_t, y_t, B_1)=c\int_{x_t}^{\infty}x^{-\gamma}dx+b\int_{y_t}^{\infty}y^{1-\beta}dy
\end{equation}

\begin{equation}\label{conditions}
\begin{aligned}
&s.t. & x_t \geq x_0, y_t\geq y_0\\
&&\frac{B_1T_0}{T_1}\geq \frac{(1-B_1)T_0}{T_2}\\
&& 0\leq B_1\leq 1
\end{aligned}
\end{equation}

where $c$ and $b$ satisfy Eq. \ref{membershipfee}, Eq. \ref{payment}. The figure below shows how the optimal $B_1, x_t, y_t$ depend on $\gamma, \beta, \lambda$.

In all experiments, we assume $\phi(x)=x^{1/2}$ and $\lambda$ is very small (with large $\lambda$ may lead the model will degenerates to one-channel case).  We see that $\gamma$ has a threshold behavior: below the threshold, the two-channel model will degenerate to one-channel case.  Above the threshold, the two-channel model has performance better than the one-channel model.
In addition, we find that if $\lambda$ is small enough and $\gamma$ exceeds the threshold, the optimal $B_1$ is only related to $\beta$ (see Fig. \ref{fig:B1_beta}).

\begin{figure}
	\centerline{\includegraphics[width=.65\textwidth]{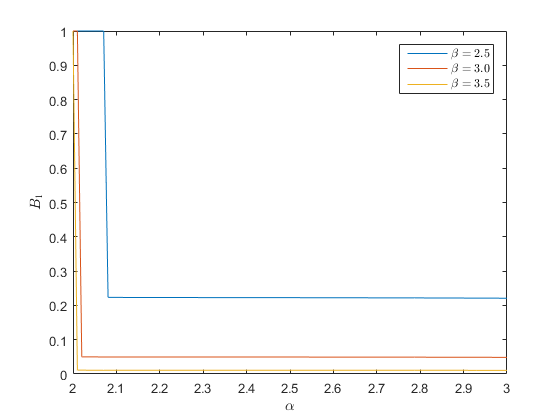}}
	\caption{Optimal $B_1$ based on different $\gamma$s and $\beta$s .}
	\label{fig:B1_alpha}
\end{figure}

\begin{figure}
	\centerline{\includegraphics[width=.65\textwidth]{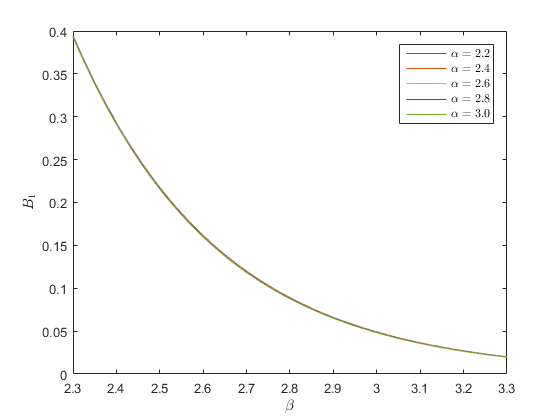}}
	\caption{Optimal $B_1$ based on different $\gamma$s and $\beta$s .}
	\label{fig:B1_beta}
\end{figure}

Finally, $x_t=x_0$ persists in the two-channel model:

\begin{theorem}\label{thm:two-channel}
	If $\phi(\cdot)$ is a increasing concave function and $\forall x\geq 0, \phi(x)\geq 0$, in the two-channel model, the optimal $x_t$ is always be $x_0$ for any $a$.
\end{theorem}

\begin{proof}
	The revenue function of two-channel platform can be represented as follows:
	\begin{equation}\label{eq:revenue_two-channel}
	\begin{aligned}
	R(x_t, y_t, B_1) &= c\int_{x_t}^{\infty}x^{-\gamma}dx+b\int_{y_t}^{\infty}y^{1-\beta}dy\\
	&=\phi((\int_{y_0}^{y_t}\frac{(1-B_1)T_0}{T_2}y^{1-\beta}dy+\int_{y_t}^{\infty}\frac{B_1T_0}{T_1}y^{1-\beta}dy)x_t)\int_{x_t}^{\infty}x^{-\gamma}dx\\
	&+\lambda\int_{x_t}^{\infty}x^{1-\gamma}dx(\frac{B_1T_0}{T_1}-\frac{(1-B_1)T_0}{T_2})\int_{y_t}^{\infty}y^{1-\beta}dy\\
	& = \phi(\frac{\sqrt{T_0}(\sqrt{\int_{y_t}^{\infty}y^{1-\beta}dy}\sqrt{B_1}+\sqrt{\int_{y_0}^{y_t}y^{1-\beta}dy}\sqrt{1-B_1})}{\sqrt{\int_{x_t}^{\infty}x^{1-\gamma}dx}}x_t)\int_{x_t}^{\infty}x^{-\gamma}dx\\
	& +\lambda\sqrt{\int_{x_t}^{\infty}x^{1-\gamma}dx}\sqrt{T_0}(\frac{\sqrt{B_1}}{\sqrt{\int_{y_t}^{\infty}y^{1-\beta}dy}}-\frac{\sqrt{1-B_1}}{\sqrt{\int_{y_0}^{y_t}y^{1-\beta}dy}})\int_{y_t}^{\infty}y^{1-\beta}dy\\
	\end{aligned}
	\end{equation}

	Then we focus on the decision problem that "Which $x_t$ is optimal?" Firstly, we let
	$$\frac{\sqrt{T_0}(\sqrt{\int_{y_t}^{\infty}y^{1-\beta}dy}\sqrt{B_1}+\sqrt{\int_{y_0}^{y_t}y^{1-\beta}dy}\sqrt{1-B_1})}{\sqrt{\int_{x_t}^{\infty}x^{1-\gamma}dx}}x_t=\eta(x_t,y_t,B_1)$$
	\begin{equation}\label{eq:partial_x_t_two-channel}
	\begin{aligned}
	\frac{\partial R(x_t, y_t, B_1)}{\partial x_t} &= -\phi(\eta)x_t^{-\gamma}+\phi'(\eta)\frac{\partial\eta}{\partial x_t}\int_{x_t}^{\infty}x^{-\gamma}dx\\
	&-\frac{1}{2}\lambda\sqrt{T_0}(\frac{\sqrt{B_1}}{\sqrt{\int_{y_t}^{\infty}y^{1-\beta}dy}}-\frac{\sqrt{1-B_1}}{\sqrt{\int_{y_0}^{y_t}y^{1-\beta}dy}})\int_{y_t}^{\infty}y^{1-\beta}dy\frac{x_t^{1-\gamma}}{\sqrt{\int_{x_t}^{\infty}x^{1-\gamma}dx}}\\
	&\leq -\phi(\eta)x_t^{-\gamma}+\phi'(\eta)(1+\frac{\gamma-2}{2})\frac{1}{\gamma-1}\eta x_t^{-\gamma}\\
	&\leq -\phi(\eta)x_t^{-\gamma}+\frac{\gamma}{2(\gamma-1)}\phi(\eta)x_t^{1-\gamma}\\
	& = (\frac{\gamma}{2(\gamma-1)}-1)\phi(\eta)x_t^{1-\gamma}\leq 0
	\end{aligned}
	\end{equation}
	
	The first inequality is based on the same calculation as Eq.\ref{eq:partial_x} where $\frac{\partial \eta}{\partial x_t}=(1+\frac{\gamma-2}{2})\frac{\eta}{x_t}$ and $\int_{x_t}^{\infty}x^{-\gamma}dx=\frac{1}{\gamma-1}x_t^{1-\gamma}$. The second inequality above is based on the fact that $\forall x\geq 0, x\phi'(x)\leq \phi(x)$.
	
	Therefore, the optimal $x_t$ for two-channel platform is always be $x_0$.
\end{proof}

\end{document}